\documentclass[draftcls,twoside,onecolumn,doublespace,journal,times,12pt]{IEEEtran}
\usepackage{graphicx}
\usepackage{amssymb}
\usepackage{amsmath}
\usepackage{cite}
\usepackage[acronym]{glossaries}
\newtheorem{lemma}{Lemma}
\newacronym{crn}{CRN}{cognitive radio network}
\newacronym{su}{SU}{secondary user}
\newacronym{su-tx}{SU-Tx}{secondary transmitter}
\newacronym{su-rx}{SU-Rx}{secondary receiver}
\newacronym{pu}{PU}{primary user}
\newacronym{pu-tx}{PU-Tx}{primary transmitter}
\newacronym{pu-rx}{PU-Rx}{primary receiver}
\newacronym{qos}{QoS}{quality of service}
\newacronym{snr}{SNR}{signal-to-noise ratio}
\newacronym{sinr}{SINR}{signal-to-interference-plus-noise ratio}
\newacronym{fcc}{FCC}{Federal Communications Commission}
\newacronym{cdf}{CDF}{cumulative distribution function}
\newacronym{pdf}{PDF}{probability density function}
\newacronym{rv}{RV}{random variable}
\newacronym{inid}{i.n.i.d.}{independent but not necessarily identically distributed}
\newacronym{eav}{EAV}{eaversdropper}
\newacronym{sep}{SEP}{symbol error probability}
\newacronym{d2d}{D2D}{device-to-device }
\newcommand{\Power}[1]{{P}_{\text{#1}}}
\newcommand{\PowerAd}{\mathcal{P}}
\newcommand{\Ppeak}{{{P}_{\text{pk}}}}
\newcommand{\Pout}[1]{{P}^{\text{#1}}_{\text{out}}}
\newcommand{\SINR}[1]{\gamma _{\text{#1}}}
  
\title{Impact of Secondary User Communication on Security Communication of Primary User }
\author{
        {
          Louis Sibomana, Hung Tran, and Quang Anh Tran
        }
\thanks{
        Louis Sibomana is with the Blekinge Institute of Technology, SE-371 79 Karlskrona, Sweden,
(e-mail: {\tt lsm@bth.se}).
         }
\thanks{
        Hung Tran is with the Faculty of Information Technology, NIEM, 31 Phan Dinh Giot, Thanh Xuan District, Hanoi, Vietnam
(e-mail: {\tt tranhungemail@gmail.com}).
         }
\thanks{
        Quang Anh Tran is with the Information Technology Faculty, Hanoi University, Vietnam
(email:\tt{anhtq@hanu.edu.vn}).}
}

\begin{document}
\maketitle
\begin{abstract}
Recently, spectrum sharing has been considered as a promising solution to improve the spectrum utilization. It however may be vulnerable to security problems as the primary and secondary network access the same resource. Therefore,
in this paper, we focus on the performance analysis of a cognitive radio network in the presence of an eavesdropper (EAV) who illegally listens to the primary user (PU) communication in which the transmit power of the secondary transmitter (SU-Tx) is subject to the joint constraint of peak transmit power of the SU-Tx and outage probability
of the PU. Accordingly, an adaptive transmit power policy and an analytical expression of symbol error probability
are derived for the SU. Most importantly, security evaluations of primary network in terms of the probability of existence of non-zero secrecy capacity and outage probability of secrecy capacity are obtained. Numerical results reveal a fact that the security of the primary network does not only depends on the channel mean powers between primary and secondary networks, but also strongly depends on the channel condition of the SU-Tx$\rightarrow$EAV link and transmit power policy of the SU-Tx.
\end{abstract}
\section{INTRODUCTION}\label{sec:1}
Recently, \gls{crn} has been considered as a feasible solution in improving the spectrum utilization\cite{Xtreport,SHK:05:IEEE_J_JSAC,Akyildiz2006} in which the spectrum underlay approach, one of spectrum access techniques in \gls{crn}, has received a lot attention in academia \cite{ASAII:09:IEEE,Khoshkholgh2010}. In particular, in underlay \gls{crn}, the \gls{su} is allowed to simultaneously access the licensed spectrum of the \gls{pu} as long as the inflicted interference at the \gls{pu} is kept below a predefined threshold \cite{SHK:05:IEEE_J_JSAC,Akyildiz2006,ASAII:09:IEEE}.
 To protect the communication of the \gls{pu}, different constraints such as outage constraint, peak or average interference power constraints, and  power allocation strategies for the \gls{su} have been studied \cite{XRYH:11:IEEE_J_JSAC,Zhang2009,Smith2013}. By doing so, the spectrum utilization have been improved significantly. Even though, spectrum underlay approach may reveal disadvantages in security issues of both secondary and primary networks. This is due to the broadcast nature of wireless signals.

  Generally, secure communication of wireless networks is typically achieved by using cryptographic protocols above the physical layer, but the signal may be decoded at the physical layer. As a consequence, information theoretic security at the physical layer has become one of the most concerned topics in wireless communication. More specifically, the problem of secure transmission  was studied from an information theoretic perspective for a wiretap channel \cite{Wyner1975}. The aim of information theoretic security is to provide a measurement that how much information can be transmitted safely by exploiting the physical characteristics of the wireless channel with the existence of \gls{eav} \cite{Barros2006,Bloch2008,Gopala2008,Dong2010}. In \cite{Bloch2008,Gopala2008}, a secrecy capacity concept has been introduced to evaluate the security level of transmitted messages as there exists an \gls{eav}, i.e., the secrecy capacity is defined as the maximum transmission rate at which a message can be reliably received by the legitimate receiver but kept perfectly secret from the \gls{eav}. This performance metric, i.e., ergodic capacity is suitable for delay tolerant applications. Otherwise, the  outage probability of secrecy capacity is suitable for delay-limited information applications \cite{Barros2006,Bloch2008}.
  In a spectrum sharing \gls{crn}, \gls{su} and \gls{pu} share the same frequency band, they may cause mutual interference as a missed power control or missed detection happen. Hence, the security in \gls{crn} becomes more challenging. In \cite{Shu2013}, authors have presented an overview on several existing security attacks to the physical layer in \gls{crn}, and then the  secrecy capacity and outage probability of secrecy capacity of the \gls{pu} have been analyzed. Most recently, an information theoretic secrecy for \gls{d2d} communication in cellular network has been considered in \cite{Yue2013}. Analytical results have illustrated that the \gls{d2d} communication is known as the interference source to the \gls{eav} which can improve the secrecy capacity for the considered system.
  In \cite{Wu2011}, authors have considered an information secrecy cooperative game where the \gls{pu} and \gls{su} cooperate and adjust their transmit powers to maximize the secrecy and information rates. In this context, the cooperation is adopted when the \gls{pu} achieves higher secrecy rate with the help of the \gls{su}. Otherwise, the \gls{pu} does not cooperate with the \gls{su}. However, \cite{Shu2013,Wu2011} did not consider the interference constraints of the \gls{su}. As a result, the \gls{qos} of primary network is not guaranteed due to the interference from the \gls{su} transmission.  Following the concerns of security for \gls{crn}, we have investigated the probability of existence of non-zero secrecy capacity for the \gls{pu} where  the \gls{su} transmit power control is subject to the maximum acceptable \gls{pu} outage constraint and the peak transmit power of the \gls{su} \cite{LSB:2013:WPMC}. However, the mathematical approach is rather complex and it is impossible to analyze further.

  To get rid of mathematical complexity in \cite{LSB:2013:WPMC} and differ from the aforementioned works, in this paper, we use another mathematic approach  to analyze the system performance and security. In particular, it is assumed that the \gls{su-tx}
     transmit power is subject to the joint constraint of the \gls{pu} outage and \gls{su} maximum transmit power limit. Accordingly, the power allocation policy for the \gls{su-tx}, the \gls{pdf}, and \gls{cdf} for the \gls{sinr} are obtained. On this basis, we do not only analyze the probability of existence of non-zero secrecy capacity for the \gls{pu} but also the outage probability of the secrecy capacity of the \gls{pu}. Furthermore, the performance of secondary network, which is subject to the interference from the \gls{pu}, in terms of the \gls{sep} is analyzed. The numerical results indicate that the probability of existence of non-zero secrecy capacity and outage probability of secrecy capacity of the \gls{pu} strongly depend on the channel conditions of the \gls{su-tx} to the \gls{eav} link and \gls{su-tx} adaptive transmit power policy. Most interestingly, the
     security of the \gls{pu} can be improved by the interference from the \gls{su-tx} to the \gls{eav}. To the best of our knowledge, there is no previous publication addressing on this problem.

    The remainder of this paper is organized as follows. In Section \ref{sec:II}, the system and channel model and assumptions are introduced. Section \ref{sec:III} derives the  \gls{su-tx} transmit power policy, the \gls{cdf} and the \gls{pdf} of the  received \gls{sinr}.  In Section \ref{sec:IV},  the \gls{su} \gls{sep}, the \gls{pu} probability of existence of non-zero secrecy capacity and outage probability of the secrecy capacity are analyzed.
    In Section \ref{sec:V}, numerical results and discussion are provided. Finally, conclusions are presented in Section \ref{sec:VI}.
\section{SYSTEM AND CHANNEL MODEL}\label{sec:II}
Let us consider a \gls{crn} model as shown in Fig.~\ref{fig:SystemModel} in which the \gls{su-tx} utilizes the licensed frequency band of the \gls{pu}  to communicate with the \gls{su-rx} receiver. There exists an \gls{eav}  who is capable of eavesdropping the signal sent by the \gls{pu-tx} by observing the channel output. Intuitively, we can see that the \gls{su-tx} and \gls{pu-tx} can cause mutual interference to the \gls{pu-rx} and \gls{su-rx}, respectively. The considered system is a typical model of the \gls{d2d} communication where the \gls{su-tx}$\rightarrow$\gls{su-rx} link  corresponds to the \gls{d2d} link while \gls{pu-tx}$\rightarrow$\gls{pu-rx} link is an instance of uplink or downlink in the cellular network \cite{AQV2014}. On the basis of Shannon theorem, the channel capacity between the \gls{pu-tx} and \gls{pu-rx} under interference caused by the \gls{su-tx} is formulated by
\begin{align}\label{eq:C_P}
    C_{P}=B\log_2 \left(1+ \SINR{P} \right)
\end{align}
where $B$ is system bandwidth and $\SINR {P}$ is the \gls{sinr} at the \gls{pu-rx}
 defined by
\begin{align}\label{eq:SINR_P}
  \SINR{P}=\frac{\Power{p} h}{\Power{s} \alpha + N_0}
\end{align}
while $\Power{p}$, $\Power{s}$, and $N_0$ are the transmit powers of the \gls{pu-tx}, \gls{su-tx}, and noise power, respectively. Further, symbols $h$ and $\alpha$ denote the channel power gains of the \gls{pu-tx}$\rightarrow$\gls{pu-rx} communication link and the \gls{su-tx}$\rightarrow$\gls{pu-rx} interference link, respectively.
Similarly, the  capacity between the \gls{su-tx} and \gls{su-rx} under interference caused by the \gls{pu-tx} can be expressed as
\begin{align}\label{eq:C_S}
    C_S =B\log_2\left(1+ \SINR{S} \right)
\end{align}
where $\SINR S$ is the \gls{sinr} at the \gls{su-rx} which is formulated as
\begin{align}\label{eq:SINR_S}
    \SINR S=\frac{ \Power{s} g }{ \Power{p}\beta +N_0}
\end{align}
here,  $g$ and $\beta$ are channel power gains for the \gls{su-tx}$\rightarrow$\gls{su-rx} communication link and \gls{pu-tx}$\rightarrow$\gls{su-rx} interference link, respectively. Due to the nature of broadcast signal in the wireless communication, the \gls{eav} may eavesdrop the transmitted information from the \gls{pu-tx} to the \gls{pu-rx}. However, the received information at the \gls{eav} is also subject to the interference caused by the
\gls{su-tx}. Thus, the capacity between the \gls{pu-tx} and \gls{eav} over the  wire-tap channel is presented as
\begin{align}\label{eq:C_E}
    C_E =B\log_2\left(1+ \SINR{E} \right)
\end{align}
where $\gamma_{E}$ is the \gls{sinr} at the \gls{eav} and defined as
\begin{align}\label{eq:SIRN_E}
    \gamma_{E}= \frac{\Power{p}  f}{ \Power {S} \varphi +N_0 }
\end{align}
In \eqref{eq:SIRN_E}, $f$ and $\varphi$ are, respectively,  channel power gains of the \gls{pu-tx}$\rightarrow$\gls{eav} and \gls{su-tx}$\rightarrow$\gls{eav} links.

In this work, the channels are assumed to be block Rayleigh fading, i.e., the channel remains constant over one time slot, and may change independently from one slot to the next. This assumption is widely accepted in realistic models for wireless communications and is applicable for severe shadowing environment where the line-of-sight does not exist such as in crowed city with many high buildings. Moreover, we denote $\Omega_X$ as the channel mean gain where $X \in \{g,h,f,\alpha,\beta,\varphi\}$, i.e., the channel mean gains are non-identical. This is reasonable since users may be located at different positions. Accordingly, the channel power gains are \gls{inid} \glspl{rv} with exponential distribution given as follows:
\begin{align}
    \label{eq:ExponentialPDF}
    f_{X}(x)=\frac{1}{\Omega}\exp\left( -\frac{x}{\Omega} \right)
    \\
    F_{X}(x)=1-\exp\left(-\frac{x}{\Omega}\right)
\end{align}
where  $f_{X}(x)$ and $F_{X}(x)$ are  the \gls{pdf}  and \gls{cdf}  of the \gls{rv} $X$, respectively.
\subsection{Spectrum sharing constraints}
Given the considered system model, the \gls{su} try to utilize the licensed frequency band of the \gls{pu} for its communication, however this operation may cause unpredictable effects to the \gls{qos} of the \gls{pu} and it is an unacceptable issue in the \gls{pu}'s view. In order to not cause harmful interference to the \gls{pu-rx}, the interference constraints given by the \gls{pu} should be established, and the \gls{su} needs to have an appropriate power allocation policy to keep interference at the \gls{pu-rx} below a predefined threshold. In the light of this idea, the interference constraint given by the \gls{pu} can be interpreted into the outage probability constraint as  \cite{XRYH:11:IEEE_J_JSAC}
\begin{align}\label{eq:outageprobability}
   P_{\text{out}}^{ P}=\Pr\{B \log_2(1+\gamma_{P}) <  r_{p} \} \le \theta_{th}
\end{align}
where $r_p$ and $\theta_{th}$ are, respectively, the primary network target transmission rate and outage probability threshold. The equation \eqref{eq:outageprobability} can interpret by words that the \gls{su-tx} is allowed to access the licensed frequency band of the \gls{pu} and cause limited interference to the \gls{pu} as long as the outage probability of the \gls{pu} capacity is kept below a predefined threshold, $\theta_{th}$. Furthermore, the  transmit power is limited in reality, thus the \gls{su-tx} transmit power is subject to additional constraint, named as the peak transmit power or maximum transmit power limit, as
\begin{align}\label{eq:PeakPowerConstraint}
  \Power{s} \le \Ppeak
\end{align}
\subsection{Performance Metrics}
Based on the transmit power constraint given by the \gls{pu}  and \gls{su}, our aim of this paper is to investigate the system performance of \gls{crn} by calculating the \gls{su} \gls{sep} and analyze the impact of the presence of the \gls{su} on the security communication of the \gls{pu}.
\subsubsection{Symbol Error Probability}
According to \cite{MKAGIC:07:IEEE_J_WCOM},   the \gls{sep} of the \gls{su} is characterized as
 \begin{align}\label{eq:SEP_Define}
    P_e=\frac{\epsilon \sqrt {\eta} }{ 2 \sqrt {\pi}} \int \limits_{0}^{\infty} F_{\SINR{S}} (\gamma) \frac{\exp\left(- \eta \gamma \right)}{\sqrt{\gamma}} d \gamma
 \end{align}
 where $\epsilon$ and $\eta$ are constants which depends on the particular modulation scheme. For example,  for $M$-phase shift keying ($M-PSK$) modulation scheme, $\epsilon=2$ and $\eta=\sin^2 (\pi/M)$.
\subsubsection{Probability of Existence of a Non-zero Secrecy Capacity}
To analyze the secure communication of the PU under the interference from the SU, we employed the instantaneous secrecy capacity. According to results reported in \cite{Barros2006},  the  secrecy capacity of the primary communication is formulated as
\begin{align}\label{eq:SecureCapacity}
        C_{\text{sec}}  = \left\{ {\begin{array}{*{20}c}
               {B\log_2(1+\SINR{P}) - B\log_2(1+\SINR{E}),~  \SINR{P} \ge \SINR{E} }\\
               {0,~~~~~~~~~~~~~~~~~~~~~~~~~~~~~~~~~~~~~~~\SINR{P} < \SINR{E} }  \\
            \end{array}} \right.
    \end{align}
Accordingly, the probability of existence of a non-zero secrecy capacity of the PU is expressed as
\begin{align}\label{eq:Nonzero_Secrecy_Capacity}
        P_{\text{ex}}=\Pr\left\{ C_{\text{sec}} >0 \right\}=\Pr\left\{\gamma_P  >\gamma_E\right\}
\end{align}
\subsubsection{Outage Probability of Secrecy Capacity}
The outage probability of secrecy capacity is defined as the probability that the instantaneous secrecy capacity is less than a  secrecy target rate $R_s>0$. Thus, the outage probability of secrecy capacity for the primary network is given by
\begin{align}\label{eq:outage_Secrecy_Capacity}
        P_{\text{out,sec}}=\Pr\left\{C_{\text{sec}} < R_s \right\}
\end{align}
According to \cite[Eq.(6)]{HRMM:12:IEEE_J_SPL}, this performance metric can be expanded by using the total probability theorem as
    \begin{align}\label{eq:OutageProb_Secrecy_Capacity}
        P_{\text{out,sec}}=\Pr\{C_{\text{sec}}<R_s| \gamma_{P} > \gamma_{E} \} \Pr\{ \gamma_{P} > \gamma_{E} \}
        + \Pr\{C_{\text{sec}}<R_s| \gamma_{P} \le \gamma_{E} \}\Pr\{ \gamma_{P} \le \gamma_{E} \}
    \end{align}
\section{STATISTICS FUNCTIONS}\label{sec:III}
In this section, we derive the power allocation policy for the \gls{su}. Thereafter,
the \gls{cdf} and \gls{pdf} for different \gls{sinr} are obtained. Let us commence by deriving the \gls{cdf} and \gls{pdf} of a function of \glspl{rv} which are important to analyze the system performance in next subsections.
\begin{lemma}\label{lemma1}
        Assuming that $a$ and $b$ are positive constants while $X_1$ and $X_2$ are independent exponentially distributed \glspl{rv} with mean values $\Omega_1$ and $\Omega_2$, respectively. A \gls{rv} $Z$ is defined by
        \begin{align}\label{eq:Z}
           Z=\frac{a X_1}{b X_2+1}
        \end{align}
        The \gls{cdf} and \gls{pdf} of $Z$ are formulated, respectively, as follows:
            \begin{align}
                \label{eq:CDF_Z}
                F_Z(z)&=1-\frac{1}{1 +z \frac{b \Omega_2}{a \Omega_1 }} \exp\left( -\frac{z}{a \Omega_1} \right)
                \\
                \label{eq:PDF_Z}
                f_Z(z)&=
                    \frac{b\Omega_2}{a \Omega_1}
                    \frac{\exp\left( -\frac{z}{a \Omega_1 } \right)}{ \left( 1+ z\frac{b \Omega_2}{a \Omega_1} \right)^2 }
                +
                    \frac{\exp\left( - \frac{z}{a \Omega_1} \right)}{a \Omega_1 \left( 1+ z\frac{b \Omega_2}{a \Omega_1} \right) }
            \end{align}
        \end{lemma}

        \begin{proof}
                    According to the probability definition, the \gls{cdf} of the \gls{rv} $Z$ can be derived by using the same approach \cite[Eq.(14)]{YHVL:08:ICC} as follows
                    \begin{align}\label{eq:CDF_Z_Define}
                        F_Z(z)=\Pr\left\{ Z <z \right\} =\int \limits_{0}^{\infty} \Pr\left\{  X_1 <  \frac{ z(b x+1 )}{a} \right\} f_{X_2} (x) dx
                    \end{align}
                    As $X_1$ and $X_2$ are independent exponentially distributed \glspl{rv}, the equation \eqref{eq:CDF_Z_Define} can be rewritten as follows
                    \begin{align}\label{eq:CDF_Z_Derive_1}
                    F_Z(z)=\int \limits_{0}^{\infty} \left\{ 1- \exp\left[-\frac{ z(b x+1 )}{a\Omega_1} \right] \right\} \frac{1}{\Omega_2} \exp\left(- \frac{x}{\Omega_2} \right)dx
                    \end{align}
                    After integration, the \gls{cdf} of $Z$ is obtained as in \eqref{eq:CDF_Z}. Then, by differentiating \eqref{eq:CDF_Z} with respect to $z$, we obtain the \gls{pdf} of $Z$ as shown in \eqref{eq:PDF_Z}.
            \end{proof}
\subsection{Power Allocation Policy of the \gls{su-tx}}
As the \gls{su} accesses the licensed frequency band of the \gls{pu}, the \gls{su-tx} must have a flexible transmit power policy to keep the interference of the \gls{pu} below a predetermined threshold. From  \eqref{eq:outageprobability}, we derive the outage probability of the \gls{pu} to withdraw the transmit power expression of the \gls{su} as
\begin{align}\label{eq:OutagePU}
    \Pout{P}=\Pr\left\{ \frac{\Power{p} h}{ \Power{s} \alpha + N_0} < \gamma_{th} \right\}
\end{align}
where $\gamma_{th}=2^{\frac{r_p}{B}}-1$. Using the Lemma \ref{lemma1}, an expression for the \gls{pu} outage probability is  presented as
\begin{align}\label{eq:ClosedFormOutageProb_PU}
    \Pout{P}=1-\frac{\Power{p}\Omega_{h}}{\gamma_{th}\Power{s}\Omega_{\alpha} +\Power{p}\Omega_{h}} \exp\left(-\frac{N_0 \gamma_{th}}{\Power{p}\Omega_{h} } \right)
\end{align}
Substituting \eqref{eq:ClosedFormOutageProb_PU} into \eqref{eq:outageprobability} and then combining with \eqref{eq:PeakPowerConstraint} yields an adaptive transmit power policy of the \gls{su-tx} as
\begin{align}\label{eq:AdaptivePowerofSU-Tx}
    \PowerAd=\min \left\{\frac{\Power{p}\Omega_{h}}{\gamma_{th} \Omega_{\alpha}}\chi ^{+}  ,\Ppeak \right\}
\end{align}
where
\begin{align}
    \chi^{+}&=\max\left\{\frac{1}{1-\theta_{th}}\exp\left(-\frac{N_0 \gamma_{th}}{\Power{p}\Omega_{h} } \right)-1,0\right\}
 \end{align}
In what follows, the \gls{su-tx} uses the power allocation policy given in \eqref{eq:AdaptivePowerofSU-Tx} to transmit the signal to the \gls{su-rx}.
\subsection{Statistics for \glspl{sinr}}
By looking into the considered performance metrics given in \eqref{eq:SEP_Define}, \eqref{eq:Nonzero_Secrecy_Capacity}, and \eqref{eq:OutageProb_Secrecy_Capacity}, we can see that the
 \gls{cdf} and \gls{pdf} for \glspl{sinr} are important functions to analyze the system performance. Therefore, we derive these functions as follows:

Using the power allocation policy given in \eqref{eq:AdaptivePowerofSU-Tx} and setting
$c=\frac{P_p}{N_0}$ and $d=\frac{\PowerAd}{N_0}$ as the \glspl{snr}, the \glspl{sinr} at the \gls{su-rx}, \gls{pu-rx}, and \gls{eav} given respectively in \eqref{eq:SINR_S}, \eqref{eq:SINR_P}, and \eqref{eq:SIRN_E} are rewritten as
\begin{align}
    \label{eq:gammaPU_Rewrite}
        \SINR {P}&= \frac{c  h}{ d \alpha +1 }
        \\
    \label{eq:gammaEAV_Rewrite}
        \SINR {E}&= \frac{c  f}{ d \varphi  +1 }
        \\
    \label{eq:gammaSU_Rewrite}
    \SINR {S} &=\frac{d g    }{c \beta  + 1}
\end{align}

\subsubsection{\gls{cdf} and \gls{pdf} of $\SINR {P}$}
Using $\emph{Lemma 1}$, the \gls{cdf} and \gls{pdf} of $\SINR {P}$ can be obtained by setting
$a=c$, $b=d$, $\Omega_1=\Omega_h$ and $\Omega_2=\Omega_{\alpha}$ as follows:
\begin{align}\label{eq:CDF_X}
    &F_{\gamma_P}(x)= 1-\frac{1}{1 + x A_0} \exp\left( -\frac{x}{B_0} \right)
    \\
    \label{eq:PDF_X}
    &f_{\gamma_P}(x)= \exp\left( -\frac{x}{B_0 } \right)\left[
            \frac{ A_0}{\left( 1+ A_0 x \right)^2 }
    +
        \frac{1}{B_0 \left( 1+ A_0  x \right) }\right]
\end{align}
where $A_0= \frac{d \Omega_{\alpha}}{c \Omega_h }$ and $\frac{1}{B_0} = \frac{1}{c\Omega_{h}}$.

\subsubsection{\gls{cdf} and \gls{pdf} of $\SINR {E}$}
Similarly, the \gls{cdf} and \gls{pdf} of $\SINR {E}$ are, respectively, obtained by setting $a=c$, $b=d$, $\Omega_1=\Omega_f$ and $\Omega_2=\Omega_{\varphi}$ as
\begin{align}
    \label{eq:CDF_Y}
    &F_{\gamma_E}(y)=
    1-\frac{1}{1 + y D_0} \exp\left( -\frac{y}{E_0} \right)
    \\
    \label{eq:PDF_Y}
    &f_{\gamma_E}(y)=\exp\left( -\frac{y}{E_0} \right)\left[
        \frac{D_0}{\left( 1+ D_0 y \right)^2 }
    +
        \frac{1}{E_0 \left( 1+ D_0 y \right) }\right]
\end{align}
where $D_0=\frac{d \Omega_{\varphi}}{c \Omega_f }$ and $\frac{1}{E_0}=\frac{1}{c \Omega_f}$.
\subsubsection{\gls{cdf} and \gls{pdf} of $\SINR {S}$}
By setting $a=d$, $b=c$, $\Omega_1=\Omega_g$ and $\Omega_2=\Omega_{\beta}$, we also obtain the \gls{cdf} and \gls{pdf} of $\SINR {S}$ as
\begin{align}
\label{eq:CDF_U}
    &F_{\gamma_S}(u)= 1-\frac{1}{1 + u F_0} \exp\left( -\frac{u}{G_0} \right)
    \\
    \label{eq:PDF_X}
    &f_{\gamma_S}(u)=\exp\left( -\frac{u}{G_0 } \right)\left[
            \frac{ F_0}{ \left( 1+ F_0 u \right)^2 }
    +
        \frac{1}{G_0 \left( 1+ F_0  u \right) }\right]
\end{align}
where $F_0= \frac{c \Omega_{\beta}}{d \Omega_{g} }$ and $\frac{1}{G_0} = \frac{1}{d \Omega_{g}}$.
\section{PERFORMANCE ANALYSIS}\label{sec:IV}
In this section, adopting the obtained transmit power policy given in \eqref{eq:AdaptivePowerofSU-Tx}, the \gls{sep} of the \gls{su}, analytical expressions for the probability of existence of non-zero secrecy capacity, and outage probability of secrecy capacity of the primary network are derived.
\subsection{Symbol Error Probability of the \gls{su}}
By substituting \eqref{eq:CDF_U} into \eqref{eq:SEP_Define}, an expression of the \gls{su} \gls{sep} can be presented as
\begin{align}\label{eq:SEP_Rewrite}
    P_e=
    \underbrace{ \frac{\epsilon \sqrt{\eta}}{2 \sqrt{\pi}} \int \limits_{0}^{\infty} \frac{\exp\left( -\eta \gamma \right)}{\sqrt{\gamma}}d\gamma} \limits_{H_1}
     -
     \underbrace{
     \frac{\epsilon \sqrt{\eta}}{2 \sqrt{\pi}} \int \limits_{0}^{\infty} \frac{1}{(1+F_0 \gamma)\sqrt{\gamma}}  \exp\left(- \frac{\gamma}{F_1} \right) d\gamma
     }\limits_{H_2}
\end{align}
where $\frac{1}{F_1}=\frac{1}{G_0} +\eta$. Moreover, using \cite[Eq. (3.361.2)]{GradshteynRyzhik2007},  $H_1$ is  given by
\begin{align}
\label{eq:H_1}
H_1&=\frac{\epsilon}{2}
\end{align}
Furthermore, by changing variable and setting $t= \gamma+\frac{1}{F_0}$, $H_{2}$ is  obtained as
\begin{align}\label{eq:H_2}
    H_2&=\frac{\epsilon \sqrt{\eta}}{2 \sqrt{\pi}} \frac{1}{F_0}\exp\left( \frac{1}{F_0 F_1} \right)\int \limits_{\frac{1}{F_0}}^{\infty} \frac{1}{t\sqrt{t-\frac{1}{F_0}}} \exp\left(- \frac{t}{F_1} \right) dt
    \nonumber\\
    &=\frac{ \epsilon}{2} \sqrt{\frac{\eta \pi }{F_0}}\exp \left(\frac{1}{F_0 F_1}\right)\left[1-
    \mathcal{Q}
     \left( \frac{1}{\sqrt{F_0 F_1} } \right)\right]
\end{align}
where \eqref{eq:H_2} is solved with the  help of \cite[Eq. (3.363.2)]{GradshteynRyzhik2007} and $\mathcal{Q}(\cdot)$ is the error function defined as $\mathcal{Q}(z)=(2/\sqrt{\pi})\int \limits_{0}^{z} \exp{(-t^2)}dt$.

As a consequence, the analytical expression of the \gls{su} \gls{sep} is given by
\begin{align}
    P_e=
    \frac{\epsilon}{2}
    -
    \frac{ \epsilon}{2} \sqrt{\frac{\eta \pi }{F_0}}\exp \left(\frac{1}{F_0 F_1}\right)\left[1-
    \mathcal{Q}
     \left( \frac{1}{\sqrt{F_0 F_1} } \right)\right]
\end{align}
\subsection{Analysis of Secure Communication of the \gls{pu}}
In this subsection, analytical expressions of the probability of existence of non-zero secrecy capacity and outage  probability of secrecy capacity of the PU are obtained.
\subsubsection{Probability of Existence of Secrecy Capacity}
According to the margin probability definition, we can derive the probability of existence of non-zero secrecy capacity for the \gls{pu} given in \eqref{eq:Nonzero_Secrecy_Capacity}  as follows:
    \begin{align}\label{eq:Prob_Nonzero_Secrecy_Capacity_Rewrite}
        P_{\text{ex}}&=1- \int \limits_{0}^{\infty} \Pr\left\{ \SINR {P}  < y\right\} f_{\SINR {E}} (y)dy = 1-\int \limits_{0}^{\infty} \frac{1}{1 + y A_0} \exp\left( -\frac{y}{B_0} \right)
        f_{\SINR {E}} (y)dy
        \nonumber\\
        &=
        D_0 \int \limits_{0}^{\infty}
        \frac{\exp\left[ -\left(\frac{1}{B_0} +\frac{1}{E_0 } \right)y \right]}
        {\left(1 + y A_0 \right)\left( 1+ D_0 y \right)^2}
        dy
        +
        \frac{1}{E_0}
        \int \limits_{0}^{\infty}
        \frac{\exp\left[ -\left(\frac{1}{B_0} +\frac{1}{E_0} \right)y \right]}
        { \left( 1+ A_0 y \right)\left( 1 + y D_0 \right) }dy
    \end{align}
where $f_{\SINR {E}} (y)$ is given in \eqref{eq:PDF_Y}.   By setting $\frac{1}{C_0}=\frac{1}{B_0} +\frac{1}{E_0}$, we can rewrite \eqref{eq:Prob_Nonzero_Secrecy_Capacity_Rewrite}  as
    \begin{align}\label{eq:Test}
         P_{\text{ex}}=
         \underbrace
        {
        D_0 \int \limits_{0}^{\infty}
        \frac{\exp\Big( -\frac{y}{C_0} \Big)}
        {\left(1 + y A_0 \right)\left( 1+ D_0 y \right)^2}
        dy
        } \limits_{I_1}
         +
         \underbrace
        {
          \frac{1}{E_0}
        \int \limits_{0}^{\infty}
        \frac{\exp\Big( - \frac{ y}{C_0} \Big)}
        { \left( 1+ A_0 y \right)\left( 1 + y D_0 \right) }dy
        } \limits_{I_2}
    \end{align}
Moreover, $I_1$ and $I_2$ can be solved as follows:
\begin{itemize}
    \item If $A_0=D_0$, the integrals $I_1$ and
    $I_2$  can be calculated with the help of \cite[Eq. (3.353.2)]{GradshteynRyzhik2007} and \cite[Eq.(3.353.3)]{GradshteynRyzhik2007}, respectively, as
        \begin{align}\label{eq:I_1_A=B}
            I_1&= D_0
            \int \limits_{0}^{\infty}
            \frac{\exp\left[ - \frac{y}{C_0} \right]}
            {\left(1 + y D_0 \right)^3}
            dy=
            \frac{C_0 D_0-1}{2C_0 D_0}+\frac{1 }{2 C_0^2 D_0^2} \exp\left(\frac{1}{C_0D_0}\right) \Gamma\left[0,\frac{1}{C_0 D_0}\right]
             \\
             \label{eq:I_2_A=B}
             I_2&=
             \frac{1}{E_0}
             \int \limits_{0}^{\infty}
            \frac{\exp\left[ - \frac{ y}{C_0} \right]}
            { \left( 1+ D_0 y \right)^2}dy
        =
        \frac{1}{E_0 D_0}+ \frac{1}{C_0 D_0^2 E_0} \exp\left({\frac{1}{C_0 D_0}}\right)   \text{Ei}\left(-\frac{1}{C_0D_0}\right)
        \end{align}
        where $\text{Ei}(z)=-\int\limits_{-z}^{\infty } \frac{e^{-t}}{t} dt$ is the exponential integral and $\Gamma\left[0,z\right]=-\text{Ei}(-z)$ for $z>0$ is the incomplete gamma function.

    \item If $A_0\neq D_0$,  $I_1$ is derived as
        \begin{align} \label{eq:I_1_A<>B}
        I_1 &=
            \underbrace{ \frac{A_0^2 D_0 }{(D_0-A_0)^2}\int \limits_{0}^{\infty} \frac{\exp\left(-\frac{y}{C_0}\right)}{ 1+A_0 y}dy
            } \limits_ {I_{11}}
            +
            \underbrace{
            \frac{D_0^2}{D_0-A_0}
            \int \limits_{0}^{\infty}
            \frac{ \exp\left(-\frac{y}{C_0}\right)}{ (1+D_0 y)^2}dy
            }
             \limits_ {I_{12}} \nonumber\\
        &
            -
        \underbrace{
            \frac{A_0D_0^2}{(D_0-A_0)^2} \int \limits_{0}^{\infty}\frac{\exp\left(-\frac{y}{C_0}\right)}{1+D_0 y}dy }
            \limits_ {I_{13}}
        \end{align}
where the integrals $I_{11}$ and $I_{13}$ are solved using \cite[Eq. (3.352.4)]{GradshteynRyzhik2007} as
        \begin{align}
        \label{eq:I_11_ClosedForm}
        I_{11}&=
        \frac{A_0 D_0}{(D_0-A_0)^2} \exp\left(\frac{1}{A_0 C_0}\right) \Gamma\left[0,\frac{1}{A_0 C_0}\right]
        \\
        \label{eq:I_13_ClosedForm}
        I_{13}&=
        \frac{A_0 D_0}{(D_0-A_0)^2} \exp\left(\frac{1}{C_0 D_0}\right) \Gamma\left[0,\frac{1}{C_0 D_0}\right]
        \end{align}
Furthermore, with the help of \cite[Eq.(3.353.3)]{GradshteynRyzhik2007}, we obtain an  expression for $I_{12}$ as
        \begin{align}
        \label{eq:I_12_ClosedForm}
            I_{12}&= \frac{D_0}{D_0-A_0} + \frac{1}{C_0 (D_0-A_0)}\exp\left(\frac{1}{C_0 D_0}\right)\text{Ei}\left[-\frac{1}{C_0 D_0}\right]
        \end{align}
In addition, when $A_0\neq D_0$, $I_2$ is calculated as
        \begin{align} \label{eq:I_2_A<>B}
        I_2&= \underbrace{\frac{A_0}{E_0(A_0-D_0)}
         \int \limits_{0}^{\infty}
         \frac{\exp\left[-\frac{y}{C_0}\right]}{1+A_0 y}dy
         }\limits_{I_{21}}
         -
         \underbrace{
         \frac{D_0}{E_0(A_0-D_0)}
          \int \limits_{0}^{\infty}
         \frac{\exp\left[-\frac{y}{C_0}\right]}{1+D_0 y}dy
         }\limits_{I_{22}}
        %
         \end{align}
where the expressions of $I_{21}$ and $I_{22}$ are obtained as
        \begin{align}
        I_{21}&=
        \frac{1}{E_0(A_0-D_0)}
        \exp\left(\frac{1}{A_0 C_0}\right) {\Gamma}\left[0,\frac{1}{A_0 C_0}\right]
        \\
        I_{22}&=
        \frac{1}{E_0(A_0-D_0)}
        \exp\left(\frac{1}{C_0 D_0}\right) {\Gamma}\left[0,\frac{1}{C_0 D_0}\right]
        \end{align}
\end{itemize}
Finally, we obtain an analytical expression of probability of existence of secrecy capacity of the \gls{pu}  as
\begin{itemize}
\item For $A_{0}=D_{0} $,
\begin{align}\label{eq:allPexist}
        P_{\text{ex}}&=\frac{C_{0}D_{0}-1}{2 C_{0}D_{0}}+\frac{1}{D_{0}E_{0}}+\frac{1}{2C_{0}^{2}D_{0}^{2}}\exp\Big(\frac{1}{C_{0}D_{0}}\Big)\Gamma\left[0,\frac{1}{C_0 D_0}\right]\nonumber\\
      &+ \frac{1}{C_{0}D_{0}^{2}E_{0}}\exp\Big(\frac{1}{C_{0}D_{0}}\Big)\text{Ei}\left(-\frac{1}{C_0D_0}\right)
 \end{align}
 \item For $A_{0}\neq D_{0} $,
\begin{align}\label{eq:allPexisttwo}
        P_{\text{ex}}&=\frac{A_{0}D_{0}}{\Big( D_{0}-A_{0}\Big)^{2}}\left\{\exp\Big(\frac{1}{A_{0}C_{0}}\Big)\Gamma\left[0,\frac{1}{A_0 C_0}\right]-\exp\Big(\frac{1}{C_{0}D_{0}}\Big)\Gamma\left[0,\frac{1}{C_0 D_0}\right]\right\}\nonumber\\
        &+\frac{D_{0}}{D_{0}-A_{0}}+\frac{1}{C_{0}\big(D_{0}-A_{0}\big)}\exp\Big(\frac{1}{C_{0}D_{0}}\Big)
        \text{Ei}\left(-\frac{1}{C_0D_0}\right)\nonumber\\
        &+\frac{1}{E_{0}\big(A_{0}-D_{0}\big)}\left\{\exp\Big(\frac{1}{A_{0}C_{0}}\Big)\Gamma\left[0,\frac{1}{A_0 C_0}\right]-\exp\Big(\frac{1}{C_{0}D_{0}}\Big)\Gamma\left[0,\frac{1}{C_0 D_0}\right]\right\}
    \end{align}
 \end{itemize}
\subsubsection{Outage Probability of Secrecy Capacity}
The probability of outage of the secrecy capacity of the \gls{pu} in \eqref{eq:OutageProb_Secrecy_Capacity} can be rewritten as
    \begin{align}\label{eq:OutageProb_Secrecy_Capacitytwo}
        P_{\text{out,sec}}&=\underbrace{\Pr\{C_{\text{sec}}<R_s| \gamma_{P} > \gamma_{E} \}} \limits _{J_1} \Pr\{ \gamma_{P} > \gamma_{E} \}\nonumber\\
        &
        + \underbrace{\Pr\{C_{\text{sec}}<R_s| \gamma_{P} \le \gamma_{E} \}\Pr\{ \gamma_{P} \le \gamma_{E} \}} \limits _{J_2}
    \end{align}
where  $\Pr\{ \gamma_{P} > \gamma_{E} \}=P_{\text{ex}}$ and $\Pr\{C_{\text{sec}}<R_s| \gamma_{P} \le \gamma_{E} \}=1$ since $R_{s}>0$. Accordingly,  $J_2$ is given by
\begin{align}
 J_2&= \Pr\{ \gamma_{P} \le \gamma_{E} \} = 1- \Pr\{ \gamma_{P} > \gamma_{E} \} =1 -P_{\text{ex}}
 \end{align}

Furthermore, we derive $J_1$ by using the Bayes's law as follows:
    \begin{align}\label{eq:J_1}
        J_1&=\Pr\left\{ \frac{1+\SINR{P}}{1+\SINR{E}} < \xi, \SINR{P} > \SINR{E} \right\}
           =\int \limits_{0}^{\infty}
           \int \limits_{y}^{\xi (1+y) -1}
           f_{\SINR{P}}(x)
            f_{\SINR {E}} (y) dxdy
           \nonumber\\
           &=\underbrace{\int \limits_{0}^{\infty} F_{\SINR{P}}\biggl(\xi (1+y) -1\biggr)f_{\SINR{E}}(y) dy } \limits_{J_{11}}
                -
            \underbrace{\int \limits_{0}^{\infty} F_{\SINR{P}}(y)f_{\SINR{E}}(y) dy}\limits_{J_{12}}
        \end{align}
where $\xi=2^{\frac{R_s}{B}}$. Substituting \eqref{eq:CDF_X}  into \eqref{eq:J_1}, we have
\begin{align}\label{eq:J_11}
    J_{11}&= \int \limits_{0}^{\infty} \left[ 1- \frac{1}{1+A_0 [\xi(1+y)-1] }  \exp\left( - \frac{\xi (1+y)-1}{B_0} \right) \right] f_{\SINR{E}}(y) dy
    \nonumber\\
    &=1-  \underbrace{\frac{\exp\left( - \frac{\xi -1}{B_0} \right)}{1+A_0(\xi -1)}
        \int \limits_{0}^{\infty}
        \frac{\exp\left( - \frac{\xi}{B_0} y \right)}
        {1+\frac{A_0 \xi }{1+A_0(\xi -1)}y }
        f_{\SINR{E}}(y) dy }\limits_{J_{111}}
\end{align}
where again  $f_{\gamma_{E}}(y)$ is given in \eqref{eq:PDF_Y} and
\begin{align}
    J_{12}&= \int \limits_{0}^{\infty} F_{\gamma_{P}}(y) f_{\gamma_{E}}(y) dy
          =\Pr\left\{ \SINR{P} < \SINR{E} \right\} = 1- P_{ex}
\end{align}
Moreover, by setting $A_1=\frac{\exp\left( - \frac{\xi -1}{B_0} \right)}{1+A_0(\xi -1)}$ and $D_1=\frac{A_0 \xi}{1+A_0(\xi -1)}$ in \eqref{eq:J_11}, we can rewrite $J_{111}$ as
\begin{align}
    J_{111} &= A_1 \int \limits_{0}^{\infty} \frac{\exp\left( - \frac{\xi}{B_0} y \right)}{1+ D_1 y}
    \left[ \frac{D_0 \exp\left( - \frac{y}{E_0}  \right)}{(1+D_0 y)^2}
    +
    \frac{\exp\left( - \frac{y}{E_0}  \right)}{E_0 (1+D_0 y)} \right] dy
    \nonumber\\
    &=
    \underbrace{A_1 D_0 \int \limits_{0}^{\infty} \frac{\exp\left( - \frac{y}{B_1}  \right)}
    {(1+ D_1 y)(1+D_0 y)^2} dy
    }\limits_{K_1}
    +
     \underbrace{
     \frac{A_1}{E_0} \int \limits_{0}^{\infty}
     \frac{\exp\left( -  \frac{y}{B_1} \right)}
    {(1+ D_1 y)(1+D_0 y)} dy
    }\limits_{K_2}
\end{align}
where $\frac{1}{B_1}=\frac{\xi}{B_0} + \frac{1}{E_0}$. Further, $K_1$ and $K_2$ can be  obtained as follows:
\begin{itemize}
    \item If $D_1=D_0$, $K_1$ and $K_2$ are  calculated with the help of \cite[Eq.(3.353.2)]{GradshteynRyzhik2007} and \cite[Eq.(3.353.3)]{GradshteynRyzhik2007}, respectively, as
        \begin{align}
            \label{eq:K_1_D0=D1}
            K_1&= D_0 A_1 \int \limits_{0}^{\infty} \frac{\exp\left( - \frac{y}{B_1} \right)}{(1+D_0 y)^3} dy
            =A_1\left\{ \frac{1}{2} -\frac{1}{2D_0 B_1} +\frac{\exp\left( \frac{1}{D_0 B_1}\right)}{2 D^2_0 B^2_1}  \Gamma\left[0,\frac{1}{D_0 B_1}\right]\right\}
            \\
            \label{eq:K_2_A0=D1}
            K_2&=\frac{A_1}{E_0}\int \limits_{0}^{\infty} \frac{\exp\left( - \frac{y}{B_1} \right)}{(1+D_1 y)^2} dy
            =
            \frac{A_1}{D_0 E_0} + \frac{A_1}{D_0^{2} E_0 B_1}\exp\left( \frac{1}{D_0 B_1} \right) E_i \left( - \frac{1}{D_0 B_1} \right)
        \end{align}
    \item If $D_1\neq D_0 $, we can obtain $K_1$ and $K_2$, respectively, as follows:
        \begin{align}
            \label{eq:K_1_D0<>D1}
            K_1
            &=D_0 A_1 \int \limits_{0}^{\infty} \frac{\exp\left( - \frac{y}{B_1} \right)}{(1+D_0 y)^2 (1+ D_1 y)}dy = K_{11} - K_{12} + K_{13}
            \\
            \label{eq:K_2_A0<>D1}
            K_2&= \frac{A_1}{E_0}\int \limits_{0}^{\infty} \frac{\exp\left( - \frac{1}{B_1} \right)}{(1+D_1 y)(1+D_0 y)}dy = K_{21} -K_{22}
         \end{align}
         where $K_{11}$, $K_{12}$, $K_{13}$, $K_{21}$, and $K_{22}$ are calculated as follows:
         \begin{align}
            \label{eq:K_11}
            K_{11}&= \frac{ A_1 D^2_0 }{(D_0 - D_1)} \int \limits_{0}^{\infty} \frac{\exp\left( - \frac{y}{B_1} \right)}{(1+D_0y)^2 }dy\nonumber\\
            &= \frac{D_0 A_1 }{D_0 -D_1}+ \frac{A_1  }{B_1(D_0 -D_1)} \exp\left( \frac{1}{D_0 B_1} \right)
            E_i\left[ - \frac{1}{D_0 B_1 } \right]
            \\
            \label{eq:K_12}
            K_{12}&=\frac{ A_1 D^2_0 D_1}{(D_0 - D_1)^2} \int \limits_{0}^{\infty} \frac{\exp\left( - \frac{y}{B_1} \right)}{1+D_0 y}dy
            = \frac{A_1 D_0 D_1 }{(D_0 - D_1)^2} \exp\left(\frac{1}{D_0 B_1 }  \right) \Gamma\left[0, \frac{1}{D_0 B_1 } \right]
           \\
           \label{eq:K_13}
            K_{13}&=\frac{ A_1 D_0 D^2_1}{(D_0 - D_1)^2} \int \limits_{0}^{\infty} \frac{\exp\left( - \frac{y}{B_1} \right)}{1+D_1 y}dy
            = \frac{A_1 D_0 D_1 }{(D_0 - D_1)^2} \exp\left( \frac{1}{B_1 D_1} \right)\Gamma\left[ - \frac{1}{B_1 D_1} \right]
            \\
            \label{eq:K_21}
            K_{21}&=\frac{D_0 A_1}{E_0 (D_0 - D_1)} \int \limits_{0}^{\infty} \frac{ \exp\left( -\frac{y}{B_1} \right) }{1+D_0 y}dy=\frac{A_1}{E_0(D_0 - D_1)} \exp\left( \frac{1}{D_0 B_1}\right) \Gamma\left[0,\frac{1}{D_0 B_1}\right]
            \\
            \label{eq:K_22}
            K_{22}&=\frac{A_1 D_1}{E_0 (D_0 - D_1)} \int \limits_{0}^{\infty} \frac{ \exp\left( -\frac{y}{B_1} \right) }{1+D_1 y}dy= \frac{A_1}{E_0 (D_0 - D_1)} \exp\left( \frac{1}{B_1 D_1} \right) \Gamma\left[ 0, \frac{1}{B_1 D_1} \right]
        \end{align}
        It is noted that $K_{11}$ is solved using  \cite[Eq.(3.353.3)]{GradshteynRyzhik2007} while $K_{12}$, $K_{13}$, $K_{21}$, $K_{22}$ are reached with the help of \cite[Eq.(3.352.4)]{GradshteynRyzhik2007}.
\end{itemize}

Then, the final expression of $P_{\text{out,sec}}$ is obtained as
\begin{itemize}
\item For $D_1= D_0 $,
\begin{align}\label{eq:finaloutone}
        P_{\text{out,sec}}&=1-\frac{A_{1}}{2}+ \frac{A_{2}}{2D_{0}B_{1}}-\frac{A_{1}}{2D_{0}^{2}B_{1}^{2}} \exp\Big(\frac{1}{D_{0}B_{1}}\Big)\Gamma\left[0,\frac{1}{D_0 B_1}\right]\nonumber\\
      & - \frac{A_{1}}{D_{0}E_{0}}-\frac{A_{1}}{D_{0}^{2}E_{0}B_{1}}\exp\Big(\frac{1}{D_{0}B_{1}}\Big)\text{Ei}\left(-\frac{1}{D_0B_1}\right)
 \end{align}
\item For $D_1\neq D_0 $,
     \begin{align}   \label{eq:finalouttwo}
        P_{\text{out,sec}}&=1-\frac{D_0 A_1 }{D_0 -D_1}- \frac{A_1  }{B_1(D_0 -D_1)} \exp\left( \frac{1}{D_0 B_1} \right)
            E_i\left[ - \frac{1}{D_0 B_1 } \right]\nonumber\\
           & + \frac{A_1 D_0 D_1 }{(D_0 - D_1)^2} \left\{
        \exp\left(\frac{1}{D_0 B_1 } \right) \Gamma\left[0, \frac{1}{D_0 B_1 } \right] -\exp\left( \frac{1}{B_1 D_1} \right)\Gamma\left[ - \frac{1}{B_1 D_1} \right]\right\}\nonumber\\
          &+
           \frac{A_1}{E_0(D_0 - D_1)} \left\{ \exp\left( \frac{1}{B_1 D_1} \right) \Gamma\left[ 0, \frac{1}{B_1 D_1} \right]-   \exp\left( \frac{1}{D_0 B_1}\right) \Gamma\left[0,\frac{1}{D_0 B_1}\right]\right\}
    \end{align}
   \end{itemize}

\section{NUMERICAL RESULTS}\label{sec:V}
In this section, the numerical results are presented to analyze the impact of primary network parameters, \gls{su} maximum transmit power limit and channel mean powers among users on the system performance. Further, we also study the effect of the presence of the \gls{su} on the primary network security.
Unless otherwise stated, the following system parameter is used for both
simulation and analysis: system bandwidth $B=5$ MHz, e.g., bandwidth of UMTS or LTE channel.
\subsection{SU SEP}
Fig. \ref{fig:SEP_VariChannelGain} illustrates the \gls{su} \gls{sep} for BPSK modulation  with different values of the \gls{su} maximum transmit \gls{snr} $\gamma_{\max}$, $\gamma_{\max}=P_{\text{pk}}/N_{0}$, and primary network setting parameters.
\begin{itemize}
    \item Case $1$: It is observed that the \gls{su} \gls{sep} decreases with respect to the increase of the \gls{pu} transmit \gls{snr}, $P_{p}/N_{0}$. This is due to the fact that when $P_{p}/N_{0}$ increases, the \gls{su-tx} transmit \gls{snr} also increases following \eqref{eq:AdaptivePowerofSU-Tx}. However, as $P_{p}/N_{0}$ increases further, e.g. $P_{p}/N_{0}>8$ dB, the \gls{su-tx} transmit \gls{snr} can not increase further as it is bounded by $\gamma_{\max}$. As a result, the  \gls{pu} transmit \gls{snr} become a strong interference source to the \gls{su} which leads to the increase of the \gls{su} \gls{sep}.
  \item Case $2$: We set $\gamma_{\max}=10$ dB, and then compare the change of the \gls{sep} to the Case $1$ where $\gamma_{\max}=15$ dB. It is easy to see that the \gls{sep} is obtained optimal value at $P_{p}/N_{0}= 2$ dB and then increase rapidly as \gls{pu} transmit \gls{snr} increases further, $P_{p}/N_{0}>2$. Clearly, the higher $\gamma_{\max}$ is, the degradation of the \gls{sep} is slower.
 \end{itemize}

 To observe the impact of the  \gls{pu} target rate $r_{p}$ and outage threshold $\theta_{th}$ on the \gls{sep}, we consider two following cases:
\begin{itemize}
    \item Case $3$: By increasing  $r_{p}=32$ Kbps (Case $1$) to $r_{p}=42$ Kbps, the \gls{su} \gls{sep} increases, i.e, the system performance decreases. This can be explained by the fact that increase of $r_{p}$ leads to
        higher \gls{sinr} at the \gls{pu-rx}. Accordingly, the \gls{su} transmit \gls{snr} must decrease
        to satisfy the \gls{pu} outage constraint, and this results in the degradation of the \gls{su} \gls{sep}.
   \item Case $4$: We compare the \gls{su} \gls{sep} with \gls{pu} outage constraint $\theta_{th}=0.03$ to
   Case $1$ with $\theta_{th}=0.01$. Clearly, the \gls{su} \gls{sep} is decreased due to the relaxing of the \gls{pu} outage constraint.
  \end{itemize}

In Fig. \ref{fig:SEP_VariPpk}, the impact of the channel mean powers  of the interference links between primary and secondary networks and \gls{pu-tx}$\rightarrow$\gls{pu-rx} link on the \gls{su} \gls{sep} are illustrated.
\begin{itemize}
    \item Cases $5, 6$ and $7$: It can be observed that the \gls{su} \gls{sep} becomes high when the channel mean powers of both \gls{su-tx}$\rightarrow$\gls{pu-rx} and \gls{pu-tx}$\rightarrow$\gls{su-rx} interference links increase. In particular, when the channel power of the SU-Tx$\to$PU-Rx link increases $\Omega_{\alpha}=0.5$ in Case $5$ to  $\Omega_{\alpha}=2$ in Case $7$, the \gls{su} \gls{sep} is high. This is due to the fact that when $\Omega_{\alpha}$ is high, the \gls{pu-rx} suffers strong interference from the \gls{su-tx}. Accordingly, the \gls{su-tx} must reduce its transmit power to guarantee the \gls{pu} outage constraint. It is also seen that by increasing the channel mean power of the \gls{pu-tx}$\rightarrow$\gls{su-rx} from   $\Omega_{\beta}=0.5$ (Case $6$)  to $\Omega_{\beta}=2$ (Case $7$), the \gls{su} \gls{sep} becomes high. In this case, the \gls{pu-tx} becomes a interference source to the \gls{su-rx} which results in the degradation of the secondary network performance.
   \item Case $8$: We can also observe that the channel mean power of the \gls{pu-tx}$\rightarrow$\gls{pu-rx} plays an important role on the secondary network performance. For instance, by increasing $\Omega_{h}=4$ (Case $7$) to $\Omega_{h}=6$ (Case $8$), the \gls{su} \gls{sep} decreases significantly. This can be explained by the fact that when $\Omega_{h}$ increases, the \gls{pu} outage probability decreases resulting in the increase of the the \gls{su} transmit \gls{snr}.
    \end{itemize}
In addition,  the \gls{su} \gls{sep} decreases as $\Omega_{\beta}$ decreases as shown in Fig. \ref{fig:SEP_VarMod} for different modulation schemes. Therefore, as expected, the secondary network performance is degraded as the channel mean powers of the interference links between primary and secondary networks become high. It can be noted that the above results in Figs. \ref{fig:SEP_VariChannelGain}, \ref{fig:SEP_VariPpk} and  \ref{fig:SEP_VarMod} are in accordance with the \gls{su} transmit power policy  given in \eqref{eq:AdaptivePowerofSU-Tx}.

\subsection{Probability of Existence of Non-zero Secrecy Capacity of the PU}
Fig. \ref{fig:nonzero1} and Fig. \ref{fig:nonzero21} illustrate the probability of existence of secrecy capacity of the \gls{pu}. We can see that this probability does not change with the increase of the \gls{pu-tx} transmit \gls{snr} for the case of identical channel mean powers  and for different values of the \gls{su} maximum transmit \gls{snr}. In fact, the probability of existence of secrecy capacity  strongly depends on the channel condition of the \gls{su-tx}$\rightarrow$\gls{eav} link. It can be observed that the primary network security is enhanced when the channel mean power of the interference link \gls{su-tx}$\to$\gls{eav} $\Omega_{\varphi}$ increases. For example, the  probability of existence of secrecy capacity  increases significantly in Fig. \ref{fig:nonzero1} by increasing $\Omega_{\varphi}=4$  to  $\Omega_{\varphi}=7, 10$ and from   $\Omega_{\varphi}=4$ to  $\Omega_{\varphi}=8$ in Fig. \ref{fig:nonzero21}, respectively. Here, the \gls{su-tx} becomes a strong interference source to the \gls{eav} which degrades the received \gls{sinr} at the \gls{eav}, and hence the primary network security becomes high. Moreover, we can see from the Fig. \ref{fig:nonzero21} that when $\Omega_{\alpha}$ decreases, the primary network security is also improved. This is can be explained by  the fact that decreasing $\Omega_{\alpha}$ results in the increase of the \gls{su-tx} transmit \gls{snr} which results in high interference to the \gls{eav}. Thus, curves in Fig. \ref{fig:nonzero1} and Fig. \ref{fig:nonzero21} show that the presence of the \gls{su} contributes significantly  to the primary network security.

\subsection{Outage Probability of Secrecy Capacity of the PU}
Fig. \ref{fig:outage1} and Fig. \ref{fig:outage22} illustrate the outage probability of secrecy capacity of the \gls{pu}.
\begin{itemize}
    \item Cases $9$ and $10$: As discussed for the probability of existence of secrecy capacity in Fig. \ref{fig:nonzero1}, it can also be observed in Fig. \ref{fig:outage1} that the outage probability of secrecy capacity does not change with the increase of the \gls{pu-tx} transmit \gls{snr} for the case of identical channels.
 \item Cases $11$ and $12$: When the channel mean power of the \gls{su-tx}$\rightarrow$\gls{eav} link increases, e.g., $\Omega_{\varphi}=8$ in both cases,  the primary network security is improved compared to  Cases $9$ ($\Omega_{\varphi}=2$) and $10$ ($\Omega_{\varphi}=4$), respectively.
    \end{itemize}
Furthermore, Fig. \ref{fig:outage22} shows that the outage probability of secrecy capacity  decreases as the channel mean power of the  \gls{su-tx}$\rightarrow$\gls{pu-rx} link decreases, $\Omega_{\alpha}=2$ to $0.5$. Again, the \gls{su-tx} transmit \gls{snr} increases due to the decrease of $\Omega_{\alpha}$ and hence the interference from the \gls{su-tx} to \gls{eav} becomes high. Therefore, results illustrated in Fig. \ref{fig:nonzero1}, Fig. \ref{fig:nonzero21}, Fig. \ref{fig:outage1}, and Fig. \ref{fig:outage22} reveal that the primary network security  strongly depends on  the channel condition of the \gls{su-tx}$\rightarrow$\gls{eav} and  \gls{su-tx} transmit power policy. In addition, the outage probability of secrecy capacity decreases as the channel mean power of the  \gls{pu-tx}$\rightarrow$\gls{pu-rx} link increases, e.g., $\Omega_{h}=4$ to $8$ with $\Omega_{\varphi}=4$ as shown in Fig. \ref{fig:outage22}.  This is expected since the \gls{pu-tx}$\rightarrow$\gls{pu-rx} link becomes better than the \gls{pu-tx}$\rightarrow$\gls{eav} link in this scenario.
\section{CONCLUSIONS}\label{sec:VI}
In this paper, we have studied the performance of a \gls{crn} under the joint constraint of the \gls{pu} outage and maximum transmit power limit of the \gls{su}. The considered model is also a typical \gls{d2d} communication model where the \gls{pu-tx}$\rightarrow$\gls{pu-rx} link is an instance of uplink or downlink of cellular network while the \gls{su-tx}$\rightarrow$\gls{su-rx} link is the instance of \gls{d2d} communication link. Accordingly, the adaptive transmit power for the \gls{su-tx} and analytical expression for the \gls{su} \gls{sep} has been derived. Further, analytical expressions of the outage probability of secrecy capacity and probability of existence of non-zero secrecy capacity of the \gls{pu} have been  obtained. In addition, the impact of the channel conditions among users, \gls{su} peak transmit power on the system performance is investigated. Most importantly, our results indicate that the primary network security strongly depends on the channel conditions of the SU-Tx$\to$EAV link and SU  transmit power policy. Also, it reveals that the presence of the \gls{su} contributes to the primary network security enhancement. The obtained results may provide valuable information to operators and system designers in a spectrum sharing \gls{crn} where the  \gls{pu} and \gls{su} can cooperate to combat the security attack.
\bibliographystyle{IEEEtran}
\bibliography{IEEEabrv,arvreference}
\clearpage
\begin{figure}[h!]
     \centerline{\includegraphics[width=0.8 \textwidth]{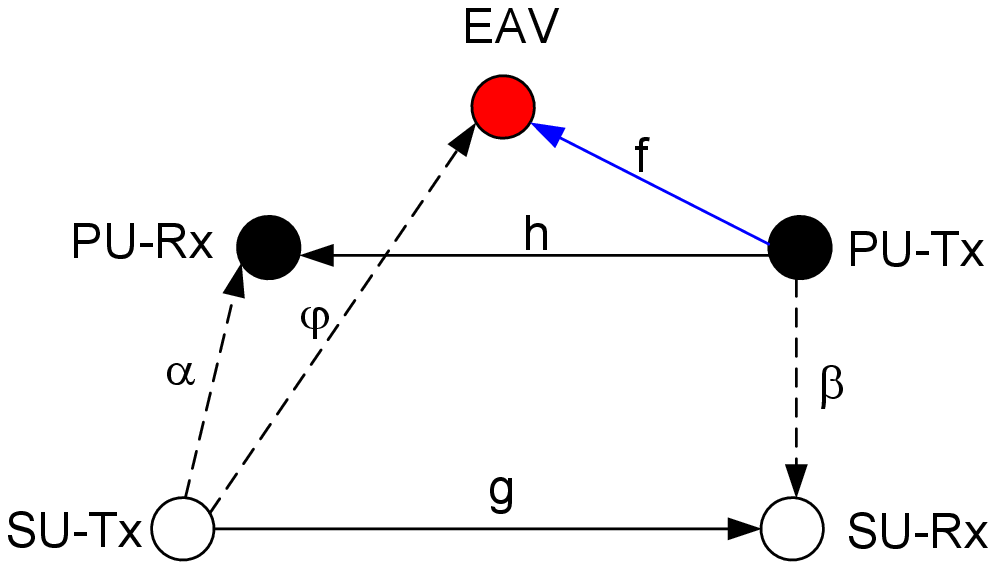}}
     \caption{
            A system model of cognitive radio network in which SU and PU share the same spectrum while an EAV  illegally  listens to the PU communication (dashed lines: Interference links; solid lines: Data information links).
        }
        \label{fig:SystemModel}
    \end{figure}
\clearpage
 \begin{figure}[h!]
     \centerline{\includegraphics[width=0.8 \textwidth]{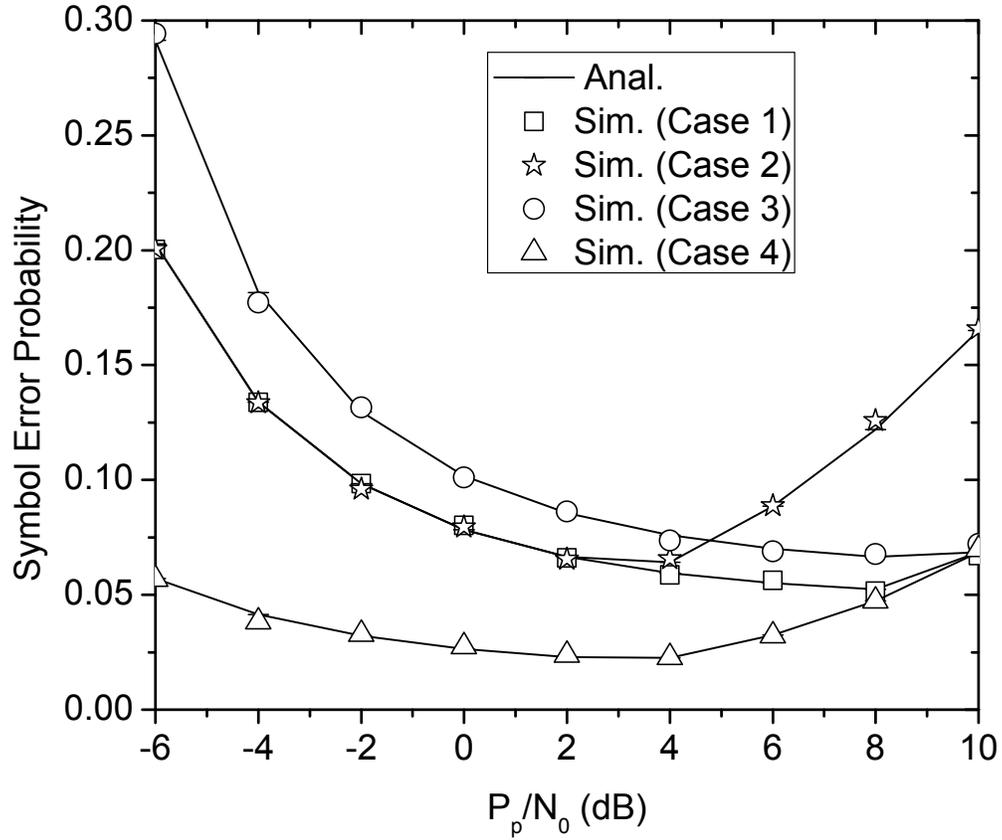}}
     \caption{SU SEP versus PU transmit SNR with BPSK modulation scheme, $\Omega_{g}=\Omega_{h}=4$ and $\Omega_{\alpha}=\Omega_{\beta}=2$. Case $1$: $\gamma_{\max}=15$ dB, $r_{p}=32$ Kbps, $\theta_{th}=0.01$; Case $2$: $\gamma_{\max}=10$ dB, $r_{p}=32$ Kbps, $\theta_{th}=0.01$;  Case $3$: $\gamma_{\max}=15$ dB, $r_{p}=42$ Kbps, $\theta_{th}=0.01$; Case $4$: $\gamma_{\max}=15$ dB, $r_{p}=32$ Kbps, $\theta_{th}=0.03$.
        }
        \label{fig:SEP_VariChannelGain}
    \end{figure}
 \clearpage
 \begin{figure}[h!]
     \centerline{\includegraphics[width=0.6 \textwidth]{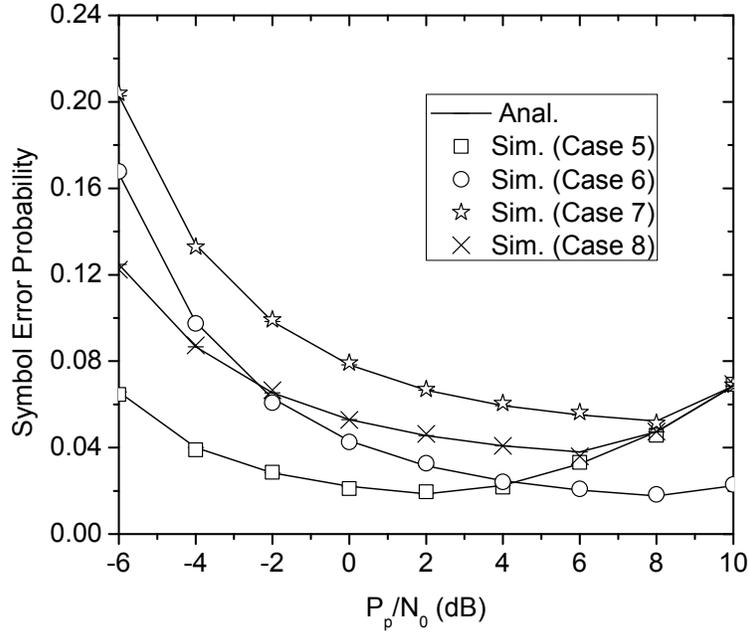}}
     \caption{SU SEP versus PU
            transmit SNR with BPSK modulation scheme, $\gamma_{\max}=15$ dB, $r_{p}=32$ Kbps,  $\theta_{th}=0.01$ and $\Omega_{g}=4$.  Case $5$: $\Omega_{h}=4$, $\Omega_{\alpha}=0.5$, $\Omega_{\beta}=2$; Case $6$: $\Omega_{h}=4$, $\Omega_{\alpha}=2$, $\Omega_{\beta}=0.5$; Case $7$: $\Omega_{h}=4$, $\Omega_{\alpha}=\Omega_{\beta}=2$; Case $8$: $\Omega_{h}=6$, $\Omega_{\alpha}=2$, $\Omega_{\beta}=2$.
        }
        \label{fig:SEP_VariPpk}
    \end{figure}
 \clearpage
 \begin{figure}[h!]
     \centerline{\includegraphics[width=0.6 \textwidth]{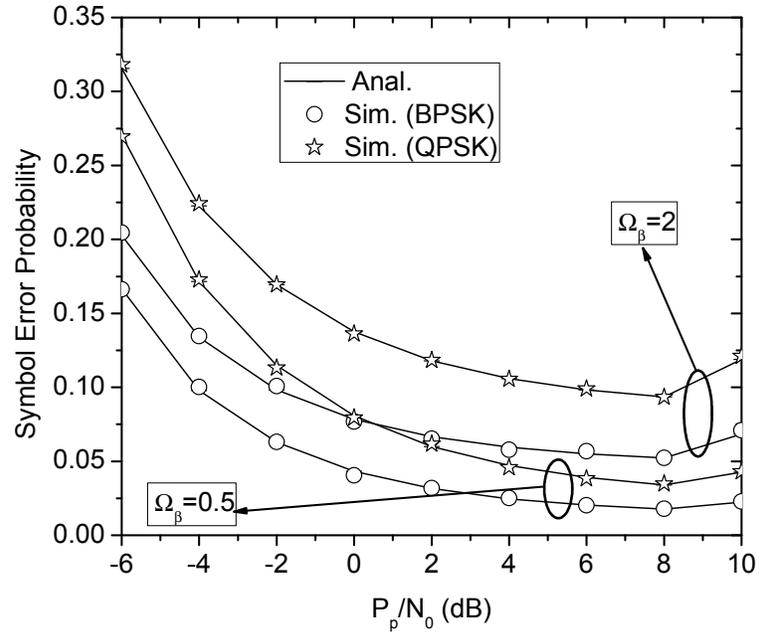}}
     \caption{SU SEP versus PU
            transmit SNR with $\gamma_{\max}=15$ dB, $r_{p}=32$ Kbps,  $\theta_{th}=0.01$,  $\Omega_{g}=\Omega_{h}=4$ and $\Omega_{\alpha}=2$.
        }
        \label{fig:SEP_VarMod}
    \end{figure}
 \clearpage
  \begin{figure}[h!]
     \centerline{\includegraphics[width=0.6 \textwidth]{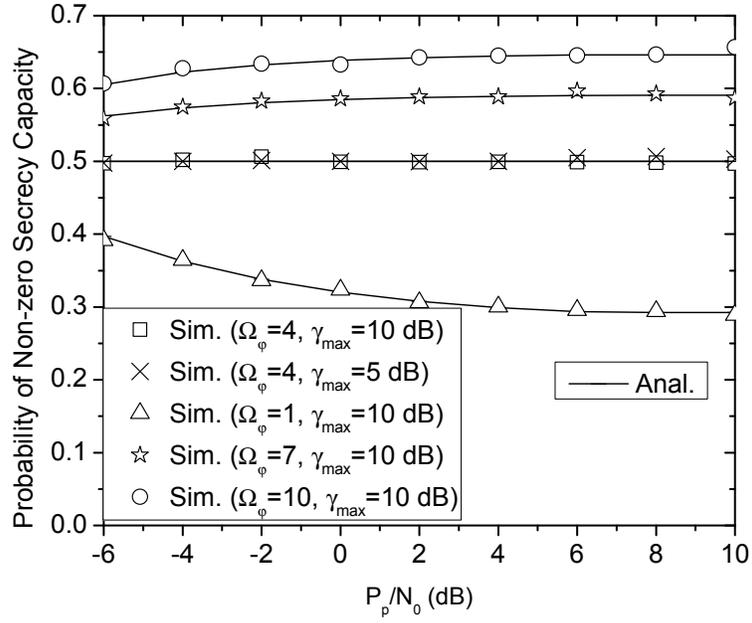}}
     \caption{Probability of existence of a non-zero secrecy capacity of the PU
            versus PU transmit SNR with  $r_{p}=32$ Kbps,  $\theta_{th}=0.01$ and $\Omega_{f}=\Omega_{g}=\Omega_{h}=\Omega_{\alpha}=\Omega_{\beta}=4$. }
        \label{fig:nonzero1}
    \end{figure}
  \clearpage
\begin{figure}[h!]
     \centerline{\includegraphics[width=0.6 \textwidth]{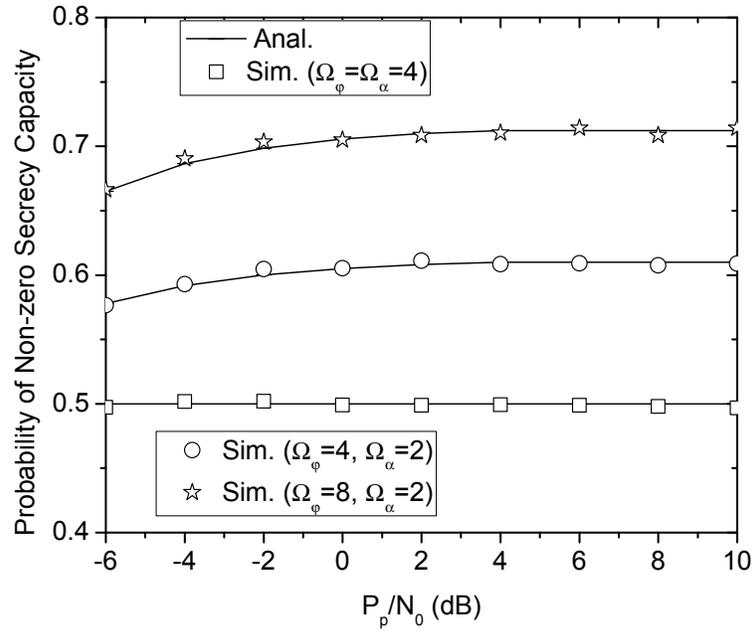}}
     \caption{Probability of existence of  non-zero secrecy capacity of the PU
            versus PU transmit SNR with $\gamma_{\max}=15$ dB, $r_{p}=32$ Kbps,  $\theta_{th}=0.01$ and $\Omega_{f}=\Omega_{g}=\Omega_{h}=\Omega_{\beta}=4$. }
        \label{fig:nonzero21}
    \end{figure}
    \clearpage
\begin{figure}[h!]
     \centerline{\includegraphics[width=0.6 \textwidth]{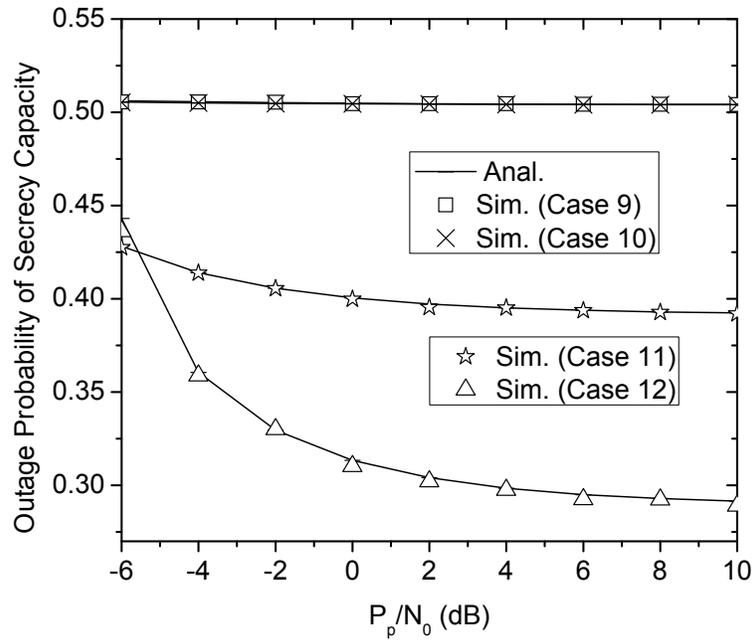}}
     \caption{Outage probability of secrecy capacity of the PU
            versus PU transmit SNR with  $R_{s}=r_{p}=32$ Kbps,  $\theta_{th}=0.01$ and $\gamma_{\max}=15$ dB.  Case $9$: $\Omega_{f}=\Omega_{h}=\Omega_{\alpha}=\Omega_{\varphi}=2$; Case $10$: $\Omega_{f}=\Omega_{h}=\Omega_{\alpha}=\Omega_{\varphi}=4$; Case $11$: $\Omega_{f}=\Omega_{h}=\Omega_{\alpha}=4$, $\Omega_{\varphi}=8$; Case $12$: $\Omega_{f}=\Omega_{h}=\Omega_{\alpha}=2$, $\Omega_{\varphi}=8$.}
        \label{fig:outage1}
    \end{figure}
    \clearpage 
\begin{figure}[h!]
     \centerline{\includegraphics[width=0.6 \textwidth]{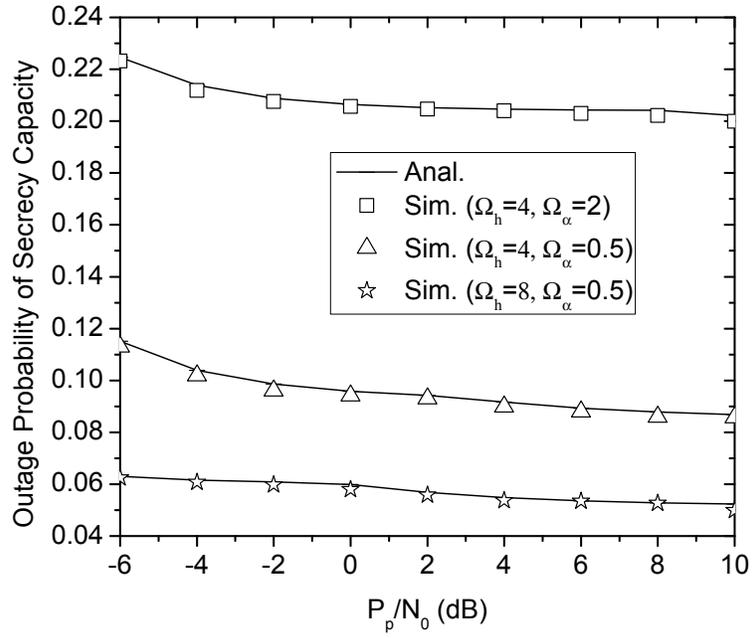}}
     \caption{Outage probability of secrecy capacity of the PU
            versus PU transmit SNR with  $R_{s}=r_{p}=32$ Kbps,  $\theta_{th}=0.01$,  $\gamma_{\max}=15$ dB and  $\Omega_{f}=\Omega_{\varphi}=4$.}
        \label{fig:outage22}
    \end{figure}
\end{document}